\documentclass[a4paper,UKenglish,cleveref, autoref, thm-restate]{lipics-v2021}
\nolinenumbers
\usepackage{amsmath}
\usepackage{amssymb}
\usepackage{amsthm}
\usepackage{dsfont}
\usepackage{color}
\usepackage{amsmath}
\usepackage{amsfonts}
\usepackage{amssymb}
\usepackage{bbm}
\usepackage{relsize}

\usepackage[procnumbered, linesnumbered,noend,algoruled]{algorithm2e}
\usepackage{graphicx}
\graphicspath{ {./images/} }
\usepackage{tikz}
\usetikzlibrary{decorations.pathreplacing,calligraphy}
\usetikzlibrary{positioning}
\usetikzlibrary {shapes,matrix}
\usetikzlibrary{arrows}
\tikzset{
	semithick,
	node distance = 2cm,
	dot/.style={circle,fill,inner sep=2pt}
}
\tikzset{
	side by side/.style 2 args={
		line width=2pt,
		#1,
		postaction={
			clip,postaction={draw,#2}
		}
	}
}
\tikzstyle{every state}=[draw = black,thick,fill = white,minimum size = 4mm]
\tikzstyle{selected edge} = [draw,line width=2pt,-,red!50]
\tikzset{
	edge/.style={->,> = latex'}
}

\newcommand{\cA}{{\mathcal{A}}}

\newcommand{\cI}{{\mathcal{I}}}
\newcommand{\cm}{{\mathcal{M}}}

\newcommand{\ck}{{\mathcal{K}}}

\newcommand{\lm}{$\ell$-matchoid~}

\newcommand{\cC}{{\mathcal{C}}}

\newcommand{\OPT}{\textnormal{OPT}}

\newcommand{\eps}{{\varepsilon}}

\title{Budgeted Matroid Maximization: $\!\!$ a Parameterized Viewpoint
}

\titlerunning{Budgeted Matroid Maximization: Parameterized Viewpoint
} 

\author{Ilan Doron-Arad}{Computer Science Department, 
	Technion, Haifa, Israel.}{idoron-arad@cs.technion.ac.il}{}{}

\author{Ariel Kulik}{CISPA Helmholtz Center for Information Security, Germany.}{ariel.kulik@cispa.de}{}{Research supported by the European Reseach Concil (ERC) consolidator grant no. 725978 SYSTEMATICGRAPH}

\author{Hadas Shachnai}{Computer Science Department, 
	Technion, Haifa, Israel.}{hadas@cs.technion.ac.il}{}{}

\authorrunning{I. Doron-Arad, A. Kulik, and H. Shachnai} 

\Copyright{Ilan Doron-Arad, Ariel Kulik, and Hadas Shachnai}



\relatedversion{} 

\EventEditors{}
\EventNoEds{0}
\EventLongTitle{}
\EventShortTitle{}
\EventAcronym{}
\EventYear{}
\EventDate{}
\EventLocation{}
\EventLogo{}
\SeriesVolume{42}
\ArticleNo{23}

\begin{document}

\maketitle

\begin{abstract}
We study budgeted variants of well known maximization problems with multiple matroid constraints. Given an $\ell$-matchoid $\cm$ on a ground set $E$, a profit function $p:E \rightarrow \mathbb{R}_{\geq 0}$, a cost function $c:E \rightarrow \mathbb{R}_{\geq 0}$, and a budget $B \in \mathbb{R}_{\geq 0}$, the goal is to find 
in the $\ell$-matchoid a feasible set $S$ of maximum profit $p(S)$ subject to the budget constraint, i.e., $c(S) \leq B$. The {\em budgeted $\ell$-matchoid} (BM) problem includes as special cases budgeted $\ell$-dimensional matching and budgeted $\ell$-matroid intersection. A strong motivation for studying BM  from parameterized viewpoint comes from the APX-hardness of unbudgeted
$\ell$-dimensional matching (i.e., $B = \infty$) already for $\ell = 3$.
Nevertheless, while there are known FPT algorithms for the unbudgeted variants of the above problems, the {\em budgeted} variants are studied here
for the first time through the lens of parameterized complexity.

We show that BM parametrized by solution size is $W[1]$-hard, already with a degenerate single matroid constraint. Thus, an exact parameterized algorithm is unlikely to exist,  motivating the study of {\em FPT-approximation schemes} (FPAS). Our main result is an FPAS for BM (implying an FPAS for $\ell$-dimensional matching and budgeted $\ell$-matroid intersection), relying on the notion of representative set $-$ a small cardinality  subset of elements which  preserves the optimum up to a small factor. We also give a lower bound on the minimum possible size of a representative set which can be computed in polynomial time. 
\end{abstract}

\section{Introduction}
\label{sec:introduction}

Numerous combinatorial optimization problems
can be interpreted as {\em constrained budgeted problems}. In this setting,  we are given a ground set $E$ of elements and a family $\cI \subseteq 2^E$ of subsets of $E$ known as the {\em feasible sets}. 
We are also given 
a cost function $c:E\rightarrow \mathbb{R}$, a profit function $p:E\rightarrow \mathbb{R}$, and a budget $B \in \mathbb{R}$. A {\em solution} is a feasible set $S \in \cI$ of bounded cost $c(S) \leq B$.\footnote{For a function $f:A \rightarrow \mathbb{R}$ and a subset of elements $C \subseteq A$, define $f(C) = \sum_{e \in C} f(e)$.} Broadly speaking, the goal is to find a solution $S$ of maximum profit.
Notable examples include budgeted matching \cite{BBGS11} and budgeted matroid intersection \cite{chekuri2011multi,grandoni2010approximation}, 
shortest weight-constrained path \cite{garey1979computers}, and constrained minimum spanning trees \cite{ravi1996constrained}. 

Despite the wide interest in constrained budgeted problems in approximation algorithms, not much is known about this intriguing
family of problems in terms of parameterized complexity. In this work, we study 
budgeted maximization with the fairly general
$\ell$-{\em dimensional matching}, $\ell$-{\em matroid intersection}, and $\ell$-{\em matchoid} constraints. 

An $\ell$-dimensional matching constraint is a set system $(E,\cI)$, where $E \subseteq U_1 \times \ldots \times U_{\ell}$ for $\ell$ sets $U_1, \ldots, U_\ell$. The feasible sets $\cI$ are all subsets $S \subseteq E$ which satisfy the following.
For any two distinct tuples $(e_1,\ldots, e_{\ell}), (f_1,\ldots, f_{\ell}) \in S$ and every $i \in [\ell]$ it holds that $e_i \neq f_i$.\footnote{For any $k \in \mathbb{N}$ let $[k] = \{1,2,\ldots,k\}$.} Informally, the input for {\em budgeted $\ell$-dimensional matching} is an $\ell$-dimensional matching constraint $(E,\cI)$, profits and costs for the elements in $E$, and a budget. The objective is to find a feasible set which maximizes the profit subject to the budget constraint (see below the formal definition). 

We now define an $\ell$-matroid intersection.
 A {\em matroid} is a set system $(E, \cI)$, where $E$ is a finite set and $\cI \subseteq 2^E$, such that 
 \begin{itemize}
	\item $\emptyset \in \cI$.
	\item The {\em hereditary property}: for all $A \in \cI$ and $B \subseteq A$ it holds that $B \in \cI$.
	\item The {\em exchange property}: for all $A,B \in \cI$ where $|A| > |B|$ there is $e \in A \setminus B$ such that $B \cup \{e\} \in \cI$. 
\end{itemize}

For a fixed $\ell \geq 1$, 
let $(E,\cI_1), (E,\cI_2), \ldots, (E,\cI_{\ell})$ be $\ell$ matroids on the same ground set $E$. An $\ell$-matroid intersection is a set system $(E,\cI)$ where 
$\cI = \cI_1 \cap \cI_2 \cap \ldots \cap \cI_{\ell}$.  
Observe that $\ell$-dimensional matching, where $E \subseteq U_1 \times \ldots \times U_{\ell}$, is a special case of $\ell$-matroid intersection: For each $i \in [\ell]$, define a partition matroid $(E,\cI_i)$, where any
feasible set $S \in \cI_i$ may contain each element $e \in U_i$ in the $i$-th coordinate at most once, i.e.,
 $$\cI_i = \{S \subseteq E~|~ \forall (e_1,\ldots, e_{\ell}) \neq (f_1,\ldots, f_{\ell}) \in S : e_i \neq f_i\}.$$
 We give an illustration in Figure~\ref{fig:Pi}. 
 It can be shown that $(E,\cI_i)$ is a matroid for all $ i \in \ell$ (see, e.g., \cite{schrijver2003combinatorial}).

 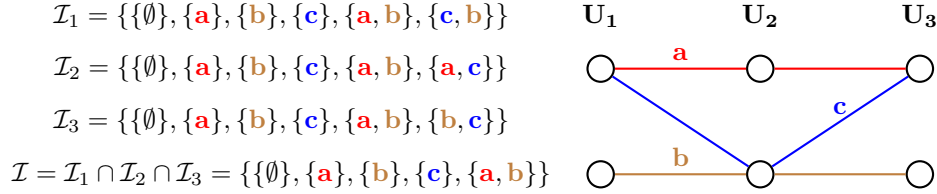
\begin{figure} 
 	\centering
 	\begin{tikzpicture}[thick, scale=1.4, every node/.style={draw, circle, inner sep=1pt}]
 	\node (00) at (4,3) {\bf ~~~}; 
 	\node (01) at (4,4)  {\bf ~~~}; 
 	\node (11) at (5.5,4)  {\bf ~~~}; 
 	\node (22) at (7,4)  {\bf ~~~}; 
 	\node (33) at (5.5,3)  {\bf ~~~}; 
 	\node (44) at (7,3)  {\bf ~~~}; 
 	\draw[color=red] (11) -- (22);
 	\draw[color=red] (01) -- node[draw=none,midway, above] {$\bf \textcolor{red}{a}$} (11);
 	\draw[color=brown] (00) --node[draw=none,midway, above] {$\bf \textcolor{brown}{b}$} (33);
 	\draw[color=brown] (33) -- (44);
 	\draw[color=blue] (01) -- (33);
 	\draw[color=blue] (33) --node[draw=none,midway, above] {$\bf \textcolor{blue}{c}$} (22);

 	\node[draw=none,ultra thick] at (4, 4.5) {$\bf \textcolor{black}{U_1}$};
 	
 	\node[draw=none,ultra thick] at (5.5, 4.5) {$\bf \textcolor{black}{U_2}$};

 	\node[draw=none,ultra thick] at (7, 4.5) {$\bf \textcolor{black}{U_3}$};
 	\node[draw=none] at (1, 4.5) {$\mathcal{I}_{1} = \left\{ \{\emptyset\}, \{\bf \textcolor{red}{a}\}, \{\bf \textcolor{brown}{b}\},\{\bf \textcolor{blue}{c}\},\{\bf \textcolor{red}{a},\bf \textcolor{brown}{b}\},  \{\bf \textcolor{blue}{c},\bf \textcolor{brown}{b}\} \right\}$};
 	
 	\node[draw=none] at (1, 4) {$\mathcal{I}_{2} = \left\{ \{\emptyset\}, \{\bf \textcolor{red}{a}\}, \{\bf \textcolor{brown}{b}\},\{\bf \textcolor{blue}{c}\},\{\bf \textcolor{red}{a},\bf \textcolor{brown}{b}\}, \{\bf \textcolor{red}{a},\bf \textcolor{blue}{c}\}\right\}$};
 	
 	\node[draw=none] at (1, 3.5) {$\mathcal{I}_{3} = \left\{ \{\emptyset\}, \{\bf \textcolor{red}{a}\}, \{\bf \textcolor{brown}{b}\},\{\bf \textcolor{blue}{c}\},\{\bf \textcolor{red}{a},\bf \textcolor{brown}{b}\}, \{\bf \textcolor{brown}{b},\bf \textcolor{blue}{c}\}\right\}$};
 	
 	\node[draw=none] at (1, 3) {$\mathcal{I} = \cI_1 \cap \cI_2 \cap \cI_3 = \left\{ \{\emptyset\}, \{\bf \textcolor{red}{a}\}, \{\bf \textcolor{brown}{b}\},\{\bf \textcolor{blue}{c}\},\{\bf \textcolor{red}{a},\bf \textcolor{brown}{b}\}\right\}$};

 	\end{tikzpicture}
 	\vspace{-3cm} 
 	\caption{\label{fig:Pi} A $3$-dimensional matching viewed as a $3$-matroid intersection. Each element in $E \subset U_1 \times U_2 \times U_3$ is represented by a path of a different color, that is $E = \{a,b,c\}$. There is a matroid constraint $(E,\cI_i)$ for each $U_i, i = 1,2,3$. The feasible sets for the matching are exactly the common independent sets of $\cI_i,i = 1,2,3$.}
 \end{figure}

The above constraint families can be generalized to the notion of $\ell$-{\em matchoid}. Informally, an \lm is  an intersection of an unbounded number of matroids, where each element belongs to at most $\ell$ of the matroids. Formally,
for any  $\ell \geq 1$, an \lm on a set $E$ is a collection $\cm = \left\{  M_i = (E_i, \cI_i) \right\}_{i \in [s]}$ of $s \in \mathbb{N}$ matroids, where for each $i \in [s]$ it holds that $E_i \subseteq E$, and every $e \in E$ belongs to at most $\ell$ sets in $\{E_1, \ldots, E_s\}$, i.e., $|\{i\in [s] ~|~ e\in E_i\}| \leq \ell$. 
A set $S \subseteq E$ is {\em feasible} for $\cm$ if for all $i \in [s]$ it holds that $S \cap E_i \in \cI_i$. Let $\cI(\cm) = \{S\subseteq E ~|~  \forall i \in [s]:~S\cap E_i\in \cI_i\}$ be all feasible sets of $\cm$. For all $k \in \mathbb{N}$, we use $\cm_k \subseteq \cI(\cm)$ to denote all feasible sets of $\cm$ of cardinality at most $k$. Clearly, $\ell$-matroid intersection (and also $\ell$-dimensional matching) is the special case of \lm where the $s (= \ell)$ matroids are defined over the same ground set $E$. 

In the {\em budgeted $\ell$-matchoid (BM)} problem, we are given an $\ell$-matchoid along with a cost function, profit function, and a budget; our goal is to maximize the profit of a feasible set under the budget constraint. The {\em budgeted $\ell$-matroid intersection (BMI)} and  {\em budgeted $\ell$-dimensional matching (BDM)} are the special cases where the \lm is an $\ell$-matroid intersection and $\ell$-dimensional matching, respectively. Each of these problems generalizes the classic $0/1$-knapsack, where all sets are feasible. Figure~\ref{fig:diagram} shows the relations between the problems. Henceforth, we focus on the BM problem.

 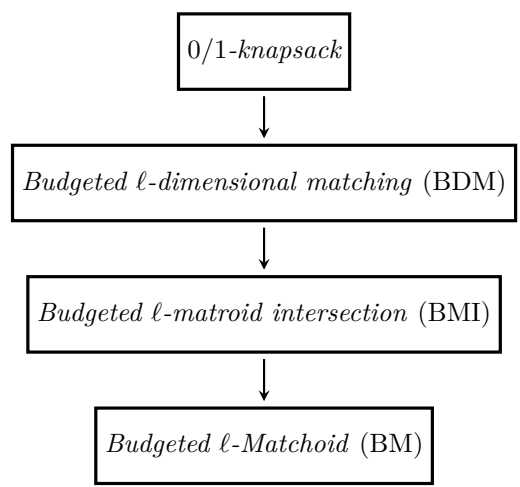
\begin{figure}\hspace*{6cm}
	\label{fig:diagram1}
	\begin{tikzpicture}[
		terminal/.style={
			rectangle,
			minimum width=2cm,
			minimum height=1cm,
			very thick,
			draw=black,
			font=\itshape,
		},
		]
		\matrix[row sep=0.7cm,column sep=-5cm] {%
			&; \node [terminal](KP) {$0/1$-knapsack};  &\\
			
			&\node [terminal](PKP) {Budgeted $\ell$-dimensional matching \textnormal{(BDM)}}; &\\
			&\node [terminal](BLM) {Budgeted $\ell$-matroid intersection \textnormal{(BMI)}}; \\
			& \node [terminal](BMI) {Budgeted $\ell$-Matchoid \textnormal{(BM)}}; \\
		};

		\draw (KP) edge [->,>=stealth,shorten <=2pt, shorten >=2pt, thick] (PKP);
		\draw (PKP) edge [->,>=stealth,shorten <=2pt, shorten >=2pt, thick] (BLM);
		\draw (BLM) edge [->,>=stealth,shorten <=2pt, shorten >=2pt, thick] (BMI);
	\end{tikzpicture}
	\caption{\label{fig:diagram} An overview of constrained budgeted problems. An arrow from problem $A$ to problem $B$ indicates that $A$ is a special case of $B$.}
\end{figure}

Formally, a BM instance is a tuple $I = (E, \cm, c,p, B,k,\ell)$, where $E$ is a ground set of elements, $\cm$ is an \lm on $E$, $c:E \rightarrow \mathbb{N}_{> 0}$ is a cost function, $p:E \rightarrow \mathbb{N}_{> 0}$ is a profit function, $B \in \mathbb{N}_{> 0}$ is a budget, and $k,\ell \in  \mathbb{N}_{> 0}$ are integer parameters.\footnote{We assume integral values for simplicity; our results can be generalized also for real values.} 
 In addition, each matroid $(E_i,\cI_i) \in \cm$ has a {\em membership oracle}, which tests whether a given subset of $E_i$ belongs to $\cI_i$ or not in a single query.
A {\em solution} of $I$ is a feasible set $S \in \cm_{k}$ such that $c(S) \leq B$. The objective is to find a solution $S$ of $I$ such that $p(S)$ is maximized. We consider algorithms parameterized by $k$ and $\ell$ (equivalently, $k+\ell$). 

We note that even with no budget constraint (i.e., $c(E)< B$), where the \lm is restricted to be a $3$-dimensional matching, BM
is  MAX SNP-complete \cite{kann1991maximum}, i.e., it cannot admit a {\em polynomial time approximation scheme (PTAS)} unless P=NP. On the other hand, the  $\ell$-dimensional matching and even the \lm problem (without a budget), parameterized by $\ell$ and the solution size $k$, are {\em fixed parameter tractable (FPT)} \cite{goyal2015deterministic,huang2023fpt}. 
This motivates our study of BM through the lens of parameterized complexity. We first observe that BM parameterized by the solution size is W[1]-hard, already with a degenerate matroid where 
all sets are feasible (i.e., knapsack parametrized by the cardinality of the solution, $k$). 

	\begin{lemma}
	\label{lem:hard}
	\textnormal{BM} is $W[1]$-\textnormal{hard}.
\end{lemma}

By the hardness result in Lemma~\ref{lem:hard}, the best we can expect for BM in terms of parametrized algorithms, is an {\em FPT-approximation scheme (FPAS)}. An FPAS with parameterization $\kappa$ for a maximization problem $\Pi$ is an algorithm whose input is an instance $I$ of $\Pi$ and an $\eps>0$, which produces a solution $S$ of $I$ of value $(1-\eps) \cdot \OPT(I)$ in time $f(\eps,\kappa(|I|)) \cdot |I|^{O(1)}$ for some computable function $f$, where $|I|$ denotes the encoding size of $I$ and $\OPT(I)$ is the optimum value of $I$. We refer the reader to \cite{marx2008parameterized,FLM20} for comprehensive surveys on parameterized approximation schemes and parameterized approximations in general. To derive an FPAS for BM, we use a small cardinality {\em representative set}, 
which is a subset of elements containing the elements of an {\em almost} optimal solution for the instance. The representative set has a cardinality depending solely on $\ell,k,\eps^{-1}$ and is constructed in FPT time. Formally,
\begin{definition}
	\label{def:REP}
	Let $I = (E, \cm, c,p, B,k,\ell)$ be a \textnormal{BM} instance, $0<\eps<\frac{1}{2}$ and $R \subseteq E$. 
	Then $R$ is a {\em representative set} of $I$ and $\eps$ if there is a solution $S$ of $I$ such that the following holds.
	\begin{enumerate}
		\item $S \subseteq R$. 
		\item $p\left(S\right) \geq (1-2\eps) \cdot \OPT(I)$.
	\end{enumerate} 
\end{definition}

We remark that Definition~\ref{def:REP} slightly resembles the definition of {\em lossy kernel} \cite{lokshtanov2017lossy}. Nonetheless, the definition of lossy kernel does not apply to problems in the oracle model, including BM (see Section~\ref{sec:discussion} for further details). 

The main technical contribution of this paper is the design of a small cardinality representative set for BM. Our representative set is constructed by forming a collection of $f(\ell, k,\eps^{-1})$ {\em profit classes}, where the elements of each profit class have roughly the same profit. Then, to construct a representative set for the instance, we define a residual problem for each profit class which enables to circumvent the budget constraint. These residual problems can be solved efficiently using a construction of~\cite{huang2023fpt}. We show that combining the solutions for the residual problems, we obtain a representative set. In the following, we use $\Tilde{O}(n)$ for $O(n \cdot \textnormal{poly}(\log (n)))$.   

 \begin{lemma}
	\label{lem:Main}
	 There is an algorithm that given a \textnormal{BM} instance $I = (E, \cm, c,p, B,k,\ell)$ and $0<\eps <\frac{1}{2}$, returns in time $|I|^{O(1)}$ a representative set $R \subseteq E$ of $I$ and $\eps$ such that $|R| = \Tilde{O}\left(\ell^{(k-1) \cdot \ell} \cdot k^2 \cdot \eps ^{-2} \right)$. 
\end{lemma}

Given a small cardinality representative set, it is easy to derive an FPAS. Specifically, using an exhaustive enumeration over the representative set as stated in Lemma~\ref{lem:Main}, we can construct the following FPAS for BM, which naturally applies also for BMI and BDM.

	\begin{lemma}
	\label{thm:FPAS}
	For any \textnormal{BM} instance $I = (E, \cm, c,p, B,k,\ell)$ and $0<\eps<\frac{1}{2}$, there is an \textnormal{FPAS}  whose running time is $|I|^{O(1)} \cdot  \Tilde{O} \left( \ell^{k^2 \cdot \ell} \cdot k^{O(k)} \cdot \eps^{-2k} \right)$. 
\end{lemma}

To complement the above construction of a representative set, we show that even for the special case of an $\ell$-dimensional matching constraint, it is unlikely that a representative set of significantly smaller cardinality can be constructed in polynomial time. The next result applies to the special case of BDM. 

 \begin{lemma}
\label{lem:DM}	
For any function $f:\mathbb{N} \rightarrow \mathbb{N}$, and  $c_1,c_2 \in \mathbb{R}$ such that $c_2-c_1<0$, there is no algorithm which finds for a
given \textnormal{BM} instance $I = (E, \cm, c,p, B,k,\ell)$ and $0<\eps<\frac{1}{2}$
a representative set of size $O \left( f(\ell)  \cdot k^{\ell-c_1} \cdot  \frac{1}{\eps^{c_2}} \right)$ of $I$ and $\eps$ in time $|I|^{O(1)}$,
  unless $ \textnormal{coNP} \subseteq \textnormal{NP} / \textnormal{poly}$.
\end{lemma}

	In the proof of Lemma~\ref{lem:DM}, we use a lower bound on the kernel size of the  {\em Perfect $3$-Dimensional Matching} ($3$-PDM) problem, due to Dell and Marx \cite{dell2012kernelization,dell2018kernelization_arxiv}.\footnote{We refer the reader e.g., to~\cite{fomin2019kernelization}, for the formal definition of kernels.}  In our hardness result, we are able to efficiently construct a kernel for $3$-PDM using a representative set for BM, already for the special case of $3$-dimensional matching constraint, uniform costs, and uniform profits. 

\subsection{Related Work}

While BM is studied here for the first time, special cases of the problem have been extensively studied from both parameterized and approximative points of view. For maximum weighted \lm without a budget constraint, Huang and Ward \cite{huang2023fpt} obtained a deterministic FPT algorithm, and algorithms for a more general problem, involving a {\em coverage function} objective rather than a linear objective. Their result differentiates the \lm problem from the matroid $\ell$-parity problem which cannot have an FPT algorithm in general matroids \cite{lovasz1980matroid,jensen1982complexity}. Interestingly, when the matroids are given a linear representation, the matroid $\ell$-parity problem admits a randomized FPT algorithm \cite{marx2009parameterized,fomin2016efficient} and a deterministic FPT algorithm \cite{lokshtanov2018deterministic}. We use a construction of \cite{huang2023fpt} as a building block of our algorithm. 

The $\ell$-dimensional $k$-matching problem (i.e., the version of the problem with no budget parametrized by $k$ and $\ell$) has received considerable attention in previous studies. Goyal et al. \cite{goyal2015deterministic} presented a deterministic FPT algorithm whose running time is ${O}^*(2.851^{(\ell-1) \cdot k})$ for the weighted version of $\ell$-dimensional $k$-matching, where ${O}^*$ is used to suppress polynomial factor in the running time. This result improves a previous result of \cite{chen2011improved}. For the unweighted version of $\ell$-dimensional $k$-matching, the state of the art is a randomized FPT algorithm with running time ${O}^*(2^{(\ell-2) \cdot k})$ \cite{bjorklund2017narrow}, improving a previous result for the problem \cite{koutis2009limits}. 


Budgeted problems are well studied in approximation algorithms. As BM is a generalization of classic $0/1$-knapsack, it is known to be NP-hard. However, while knapsack admits a {\em fully PTAS (FPTAS)} \cite{martello1990knapsack}, BM is unlikely to admit a PTAS, even for the special case of $3$-dimensional matching with no budget constraint \cite{kann1991maximum}. Consequently, there has been extensive research work to identify special cases of BM which admit approximation schemes. 

For the budgeted matroid independent set (i.e., the special case of BM where the \lm consists of a single matroid), Doron-Arad et al.~\cite{DKS23} developed an {\em efficient PTAS (EPTAS)} using the representative set based technique. This algorithm was later generalized in~\cite{doron2023eptas} to tackle budgeted matroid intersection and budgeted matching (both are special cases of BM where $\ell = 2$), improving upon a result of Berger et al.~\cite{BBGS11}. We generalize some of the technical ideas of  \cite{DKS23,doron2023eptas} to the setting of $\ell$-matchoid and parametrized approximations.

\noindent
{\bf Organization of the paper:} Section~\ref{sec:RepSet} describes our construction of a representative set. In Section~\ref{sec:FPAS} we present our FPAS for BM.  Section~\ref{sec:hard} contains the proofs of the hardness results given in Lemma~\ref{lem:hard} and in Lemma~\ref{lem:DM}. In Section~\ref{sec:proofs} we present an auxiliary approximation algorithm for BM. 
We conclude in Section~\ref{sec:discussion} with a summary and some directions for future work.

\section{Representative Set}
\label{sec:RepSet}

In this section we construct a representative set for BM. 
Our first step is to round the profits of a given instance, and to determine the low profit elements that can be discarded without incurring significant loss of profit. 
We find a small cardinality representative set from which an almost optimal solution can be selected via enumeration yielding an FPAS (see Section~\ref{sec:FPAS}). 

We proceed to construct a representative set whose cardinality depends only on $\eps^{-1},k$, and $\ell$. This requires the definition of {\em profit classes}, namely, a partition of the elements into groups, where the elements in each group have 
similar profits. Constructing a representative set using this method requires an approximation of the optimum value of the input BM instance $I$. To this end, we use a $\frac{1}{2 \ell}$-approximation $\alpha = \textnormal{\textsf{ApproxBM}}(I)$ of the optimum value $\OPT(I)$ described below. 
\begin{lemma}
\label{lem:CA}
Given a \textnormal{BM} instance $I = (E, \cm, c,p, B,k,\ell)$,
there is an algorithm $\textnormal{\textsf{ApproxBM}}$ which 
 returns in time $|I|^{O(1)}$ a value $\alpha$ such that $\frac{\OPT(I)}{2\ell}\leq \alpha \leq \OPT(I) $. 
\end{lemma} The proof 
of Lemma~\ref{lem:CA} is given in Section~\ref{sec:proofs}. The proof utilizes a known approximation algorithm for the unbudgeted version of BM~\cite{jenkyns1976efficacy,jukna2011extremal} which is then transformed into an approximation algorithm for BM using a technique of~\cite{kulik2021lagrangian}.  
 
 The first step in designing the profit classes is to determine a set of {\em profitable} elements. 
 required for obtaining an almost optimal solution.
 This set allows us to 	

construct only a small number of profit classes. We define the set of
	{\em profitable} elements w.r.t. $I, \alpha$, and $\eps$ as 
\begin{equation}
	\label{eq:alphaP}
	H[I,\alpha,\eps] = \left\{ e \in E~|~ 
	 \frac{\eps \cdot \alpha}{k} < p(e) \leq 2\cdot \ell \cdot \alpha    \right\}. 
\end{equation} 
When clear from the context, we simply use $H = H[I,\alpha,\eps]$. 
Consider the non-profitable elements. The next lemma states that omitting these elements indeed has small effect on the profit of the solution set.

	\begin{lemma}
	\label{lem:A1}
	For every \textnormal{BM} instance $I = (E, \cC, c,p, B,k,\ell)$, $\frac{\OPT(I)}{2\ell} \leq \alpha \leq \OPT(I) $, $0<\eps<\frac{1}{2}$, and $S \in \cm_{k}$ it holds that $p \left(  S \setminus H[I,\alpha,\eps] \right) \leq  \eps \cdot \OPT(I)$.  
\end{lemma}
\begin{proof}
	We note that
	$$p \left(  S \setminus H[I,\alpha(I),\eps] \right) \leq k \cdot \frac{\eps \cdot \alpha}{k} = \eps \cdot \alpha \leq  \eps \cdot \OPT(I).$$
The first inequality holds since each element in $S \setminus H[I,\alpha(I),\eps]$ has profit at most $\frac{\eps \cdot \alpha}{k}$ by \eqref{eq:alphaP}; in addition, since $S \in \cm_{k}$ it follows that $S$ contains at most $k$ elements. The second inequality  holds as $\alpha \leq \OPT(I)$. 
\end{proof}

Using Lemma~\ref{lem:A1}, our representative set can be constructed exclusively from profitable elements. 
We can now partition the profitable elements into a small number of {\em profit classes}. There is a profit class $r$ for a suitable range of profit values.  
Specifically, let 
\begin{equation}
\label{eq:def_D}
D(I,\eps) = \left\{r \in \mathbb{N}_{>0}~\big|~ (1-\eps)^{r-1} \geq \frac{\eps}{2 \cdot \ell \cdot k}\right\},
\end{equation}
 and we simplify by $ D = D(I,\eps)$. 
 For all $r \in D$, and $\frac{\OPT(I)}{2\ell} \leq \alpha \leq \OPT(I) $, define the $r$-{\em profit class} as 
\begin{equation}
	\label{Er}
	{\ck}_{r}(\alpha)
	= \left\{e \in E~\bigg|~ \frac{ p(e)}{2 \cdot \ell \cdot \alpha} \in \big( (1-\eps)^{r}, (1-\eps)^{r-1} \big]\right\}.
\end{equation} 
In words, each profit class $r \in D$ contains profitable elements (and may contain some elements that are {\em almost} profitable due to our $\frac{1}{2\ell}$-approximation for $\OPT(I)$), where the profits of any two elements that belong to the $r$-profit class can differ by at most a multiplicative factor of $(1-\eps)$.  We use the following simple upper bound on the number of profit classes. 

	\begin{lemma}
	\label{lem:ProfitBound}
		For every \textnormal{BM} instance $I$ and $0<\eps<\frac{1}{2}$ there are $O( k \cdot \ell \cdot \eps^{-2})$ profit classes.  
\end{lemma}

\begin{proof}
We note that
		\begin{equation}
			\label{eq:ing}
			\log_{1-\eps} \left(\frac{\eps}{2 \ell \cdot k}\right)  \leq 
			\frac{\ln \left(\frac{2 \ell \cdot k}{\eps}\right)}{-\ln \left(1-\eps \right)} \leq \frac{ 2 \ell \cdot k \cdot \eps^{-1}}{\eps}.
		\end{equation} 
	The second inequality follows from  $x< -\ln (1-x), \forall x>-1, x \neq 0$, and $\ln (y) < y, \forall y>0$. 
By \eqref{eq:def_D}
the number of profit classes is bounded by
\begin{equation}
\label{eq:upper_bound_D}
|D| \leq \log_{1-\eps} \left( \frac{\eps}{ 2 \ell \cdot k}  \right)+1 = O( k \cdot \ell \cdot \eps^{-2}).
\end{equation}
The last inequality follows from \eqref{eq:ing}. 
\end{proof}
Next, we define an {\em exchange set} for each profit class.  This facilitates the construction of a representative set.   
Intuitively, a subset of elements $X$ forms an exchange set for a profit class ${\ck}_r(\alpha)$ if any feasible set $\Delta$ and element $a \in (\Delta \cap {\ck}_r(\alpha)) \setminus X$ can be replaced (while maintaining feasibility) by some element $b \in (X \cap {\ck}_r(\alpha)) \setminus \Delta$ such that the cost of $b$ is upper bounded by the cost of $a$. Formally, 
\begin{definition}
	\label{def:r-set}
	Let $I = (E, \cm, c,p, B,k,\ell)$ be a \textnormal{BM} instance, $0<\eps<\frac{1}{2}$, $\frac{\OPT(I)}{2\ell} \leq \alpha \leq \OPT(I) $, $r \in D (I,\eps)$, and $X \subseteq {\ck}_r(\alpha)$. We say that $X$ is an {\em exchange set} for $I,\eps,\alpha,$ and $r$ if: 
	\begin{itemize}
		\item For all $\Delta \in \cm_{k}$ and $a \in (\Delta \cap {\ck}_r(\alpha)) \setminus X$ there is $b \in ({\ck}_r(\alpha) \cap X) \setminus \Delta$ satisfying 
		\begin{itemize}
			\item $c(b) \leq c(a)$.
			\item $\Delta-a+b \in \cm_{k}$.
		\end{itemize} 
	\end{itemize}

\end{definition}

The key argument in this section is that if a set $R \subseteq E$ satisfies that $R \cap {\ck}_r(\alpha)$ is an exchange set for any $r \in D$, then $R$ is a representative set. This allows us to construct a representative set using a union of disjoint exchange sets, one for each profit class. We give an illustration in Figure~\ref{fig:2}.

%
%
%
%
%

 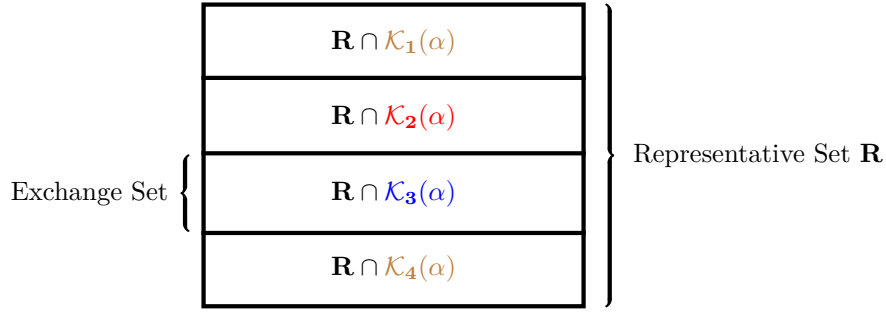
\begin{figure} 
\centering
\begin{tikzpicture}[ultra thick]
	\node[draw, rectangle, minimum width=5cm, minimum height=4cm] (R) at (0,0) {};
	
	\node[draw=none] at (4.8, 0) {Representative Set $\bf R$};
	
	\node[draw=none] at (-4, -0.5) {Exchange Set}; 
	
	\node[draw=none] at (0, 1.5) {$\bf {R \cap  \textcolor{brown}{\mathcal{K}_1(\alpha)}}$};
	\node[draw=none] at (0, 0.5) {$\bf R \cap  \textcolor{red}{\mathcal{K}_2(\alpha)}$};
	\node[draw=none] at (0, -0.5) {$\bf R \cap \textcolor{blue}{\mathcal{K}_3(\alpha)}$};
	\node[draw=none] at (0, -1.5) {$\bf R \cap  \textcolor{brown}{\mathcal{K}_4(\alpha)}$};
	
	\draw ([yshift=-1cm]R.north west) -- ([yshift=-1cm]R.north east) 
	;
	\draw ([yshift=-2cm]R.north west) -- ([yshift=-2cm]R.north east) 
	;
	\draw ([yshift=+1cm]R.south west) -- ([yshift=+1cm]R.south east) 
	;
	\draw [decorate,
	decoration = {calligraphic brace}] (-2.7,-1) --  (-2.7,0);
	
	\draw [decorate,
	decoration = {calligraphic brace}] (2.8,2) --  (2.8,-2);
\end{tikzpicture}
	\caption{\label{fig:2} An illustration of our construction of a representative set $R$, using a union of exchange sets, one for each profit class $\mathcal{K}_1(\alpha)$, $\mathcal{K}_2(\alpha)$, $\mathcal{K}_3(\alpha)$, $\mathcal{K}_4(\alpha)$. }
\end{figure}

%
%

\begin{lemma}
	\label{lem:sufficientRep}
	Let $I = (E, \cm, c,p, B,k,\ell)$ be a \textnormal{BM} instance, $0<\eps<\frac{1}{2}$, $\frac{\OPT(I)}{2\ell} \leq \alpha \leq \OPT(I) $, and $R \subseteq E$. If for all $r \in D = D(I,\eps)$ it holds that $R \cap {\ck}_r(\alpha)$ is an  exchange set for $I,\eps,\alpha,$ and $r$, then $R$ is a representative set of $I$ and $\eps$. 
\end{lemma}


For the proof of Lemma~\ref{lem:sufficientRep}, we define a {\em substitution} of some feasible set $G \in \cm_{k}$. We will use $G$ later only as an optimal solution; however, we can state the following claims for a general $G \in \cm_{k}$.  
We require that a substitution preserves the number of profitable elements in $G$ from each profit class, so a substitution guarantees a profit similar to the profit of $G$.
\begin{definition}
	\label{def:sub}
	For $G \in \cm_{k}$ and $Z_G \subseteq \bigcup_{r \in D} {\ck}_r(\alpha)$, we say that $Z_G$ is a {\em substitution} of $G$ if the following holds. \begin{enumerate}

		\item $Z_G \in \cm_{k}$.\label{pp:1}
		
		\item  $c(Z_G) \leq c(G)$.\label{pp:2}
		
		\item For all $r \in D$ it holds that $|{\ck}_r(\alpha) \cap Z_G| = |{\ck}_r(\alpha) \cap G|$.\label{pp:3}

		%
	\end{enumerate}
\end{definition}

	
\noindent {\bf Proof of Lemma~\ref{lem:sufficientRep}:}
We first show that every set $G\in \cm_k$ has a substitution which is a subset of $R$. 
\begin{claim}
	\label{claim:substitution}
	For any $G \in \cm_{k}$ there is a substitution $Z_G$ of $G$ such that $Z_G \subseteq R$. 
\end{claim}
\begin{claimproof}
	Let $G \in \cm_{k}$ and let $Z_G$ be a substitution of $G$ such that $|Z_G \cap R|$ is maximal among all substitutions of $G$; formally, let $\mathcal{S}(G)$ be all substitutions of $G$ and let $$Z_G \in \{ Z \in \mathcal{S}(G)~|~ |Z \cap R| = \max_{Z' \in \mathcal{S}(G)} |Z' \cap R|\}.$$ Since $G \cap \bigcup_{r \in D} {\ck}_r(\alpha)$ 
	is in particular a substitution of $G$ it follows that $\mathcal{S}(G)\neq \emptyset$; thus,  $Z_G$ is well defined.  Assume towards a contradiction that there is $a \in Z_G \setminus R$; then, by Definition~\ref{def:sub} there is $r \in D$ such that $a \in {\ck}_r(\alpha)$. 
	Because $R \cap {\ck}_r(\alpha)$ is an exchange set for $I,\eps,\alpha$, and $r$, by Definition~\ref{def:r-set} there is $b \in ({\ck}_r(\alpha) \cap R) \setminus Z_G$ such that $c(b) \leq c(a)$ and $Z_G -a+b \in \cm_{k}$. Then,  the properties of Definition~\ref{def:sub} are satisfied for $Z_G-a+b$ by the following. 
	\begin{enumerate}
		\item $Z_G -a +b \in \cm_{k}$ by the definition of $b$. 
		
		\item  $c(Z_G-a+b) \leq c(Z_G) \leq c(G)$ because $c(b) \leq c(a)$. 
		
		\item for all $r' \in D$ it holds that $|{\ck}_{r'}(\alpha) \cap (Z_G-a+b)| = |{\ck}_{r'}(\alpha) \cap Z_G| = |{\ck}_{r'}(\alpha) \cap G|$ because $a,b \in {\ck}_r(\alpha)$.
		
	\end{enumerate}
	
	By the above, and using and Definition~\ref{def:sub}, we have that $Z_G+a-b$ is a substitution of $G$; that is, $Z_G+a-b \in \mathcal{S}(G)$. Moreover, \begin{equation}
		\label{eq:ZG}
		|R \cap (Z_G -a+b)|>|R \cap Z_G| = \max_{Z \in \mathcal{S}(G)} |Z \cap R|.
	\end{equation} The first inequality holds since $a \in Z_G \setminus R$ and $b \in R$. Thus, we have found a substitution of $G$ which contains more elements in $R$ than $Z_G \in \mathcal{S}(G)$. A contradiction to the definition of $Z_G$ as a substitution of $G$ having a maximum number of elements in $R$. Hence, 	
	$Z_G \subseteq R$, as required.
\end{claimproof}

Let $G$ be an optimal solution for $I$. We complete the proof of Lemma~\ref{lem:sufficientRep} by showing that a substitution of $G$ which is a subset of $R$ yields a  profit at least $(1-2\eps) \cdot \OPT(I)$. Let $H[I,\alpha,\eps] = H$ be the set of profitable elements w.r.t. $I, \alpha$ and $\eps$ (as defined in \eqref{eq:alphaP}). By Claim~\ref{claim:substitution}, as $G \in \cm_{k}$, it has a substitution $Z_G \subseteq R$. Then, 
\begin{equation}
	\label{eq:profitR}
	\begin{aligned}
		p(Z_G) \geq{} & \sum_{r \in D} p({\ck}_r(\alpha) \cap Z_G) 
		\\ \geq{} &  \sum_{r \in D \text{ s.t. } {\ck}_r(\alpha) \neq \emptyset} |{\ck}_r(\alpha) \cap Z_G| \cdot \min_{e \in {\ck}_r(\alpha)} p(e) 
		\\ \geq{} & \sum_{r \in D \text{ s.t. } {\ck}_r(\alpha) \neq \emptyset} |{\ck}_r(\alpha) \cap G | \cdot (1-\eps) \cdot \max_{e \in {\ck}_r(\alpha)} p(e) 
		\\ \geq{} & (1-\eps) \cdot p(G \cap H).
	\end{aligned}
\end{equation}  The third inequality follows from \eqref{Er}, and from Property~\ref{pp:3} in Definition~\ref{def:sub}. The last inequality holds since for every $e \in H$ there is $r \in D$ such that $e \in {\ck}_r(\alpha)$, by \eqref{eq:alphaP} and \eqref{Er}. 
Therefore,  
\begin{equation}
	\label{eq:profitFINAL}
	\begin{aligned}
		p(Z_G)  \geq{} & (1-\eps) \cdot p(G \cap H) \\
		={} & (1-\eps) \cdot \left( p(G) - p(G \setminus H) \right) \\
		 \geq{} & (1-\eps) \cdot p(G)-p(G \setminus H) \\
		 \geq{} & (1-\eps) \cdot p(G)- \eps \cdot \OPT(I)\\
		  ={} & (1-\eps) \cdot \OPT(I)- \eps \cdot \OPT(I)\\
		   ={} & (1-2\eps) \cdot \OPT(I).
	\end{aligned}
\end{equation}
The first inequality follows from \eqref{eq:profitR}. The last inequality holds by Lemma~\ref{lem:A1}. The second equality holds since $G$ is an optimal solution for $I$. To conclude, by Properties~\ref{pp:1} and~\ref{pp:2} in Definition~\ref{def:sub}, it holds that $Z_G \in \cm_{k}$, and $c \left( Z_G \right) \leq c(G) \leq B$; thus, $Z_G$ is a a solution for $I$. Also, by \eqref{eq:profitFINAL}, it holds that $p \left( Z_G \right) \geq (1-2\eps) \cdot \OPT(I)$ as required (see Definition~\ref{def:REP}). \qed


By Lemma~\ref{lem:sufficientRep}, our end goal of constructing a representative set is reduced to efficiently finding exchange sets for all profit classes. This can be achieved by the following result, which is a direct consequence of Theorem 3.6 in \cite{huang2023fpt}.\footnote{The result of~\cite{huang2023fpt} refers to a maximization version of exchange sets; however, the same construction and proof hold for our exchange sets as well.} 
\begin{lemma}
	\label{lem:mainMatroid}
	Given a \textnormal{BM} instance $I = (E, \cm, c,p, B,k,\ell)$, $0<\eps <\frac{1}{2}$, $\frac{\OPT(I)}{2\ell} \leq \alpha \leq \OPT(I) $, and $r \in D (I,\eps)$, there is an algorithm $\textnormal{\textsf{ExSet}}$ which
	returns in time  $\Tilde{O}\left(\ell^{(k-1) \cdot \ell} \cdot k\right) \cdot |I|^{O(1)}$ an exchange set $X$ for $I,\eps,\alpha,$ and $r$, such that $|X| = \Tilde{O} \left(  {\ell^{(k-1) \cdot \ell} \cdot k} \right)$. 
\end{lemma}

\begin{algorithm}[h]
	\caption{$\textsf{RepSet}(I = (E, \cm, c,p, B,k,\ell),\eps)$}
	\label{alg:findRep}
	\SetKwInOut{Input}{input}
	\SetKwInOut{Output}{output}
	
	\Input{A BM instance $I$, and an error parameter $0<\eps<\frac{1}{2}$. }
	
	\Output{A representative set $R$ of $I$ and $\eps$.}
	
%
	
%

	\If{$\ell^{(k-1) \cdot \ell } \cdot k^2  \cdot \eps^{-2}> |I|$}{Return $E$ \label{step:IF}}
	
		Compute $\alpha \leftarrow \textsf{ApproxBM}(I)$.\label{step:A}
	
	\For{$r \in D (I,\eps)$\label{step:REPforr}}{
		
	$R \leftarrow R \cup \textsf{ExSet}(I,\eps,\alpha,r)$\label{step:match}. 
	
	}
	
	Return $R$ \label{step:REPret}
\end{algorithm}

Using Lemmas~\ref{lem:sufficientRep} and~\ref{lem:mainMatroid}, a representative set of $I$ can be constructed as follows. 
If the parameters $\ell$ and $k$ are too high w.r.t. $|I|$, return the trivial representative set $E$ in polynomial time. Otherwise, compute an approximation for $\OPT(I)$, and define the profit classes. Then, the representative set is constructed by finding an exchange set for each profit class. The pseudocode of the algorithm is given in Algorithm~\ref{alg:findRep}.

 \begin{lemma}
	\label{lem:main}
	Given a \textnormal{BM} instance $I = (E, \cm, c,p, B,k,\ell)$, and $0<\eps <\frac{1}{2}$, Algorithm ~\ref{alg:findRep}  returns in time 
	$|I|^{O(1)} $ a representative set $R \subseteq E$ of $I$ and $\eps$ such that $|R| = \Tilde{O}\left(\ell^{(k-1) \cdot \ell} \cdot k^2 \cdot \eps ^{-2} \right)$. 
\end{lemma}

\begin{proof}
Clearly, if $\ell^{(k-1) \cdot \ell } \cdot k^2 \cdot \eps^{-2} > |I|$, then by Step~\ref{step:IF} the algorithm runs in time $|I|^{O(1)}$ and returns the trivial representative set $E$. Thus, we may assume below that
$\ell^{(k-1) \cdot \ell } \cdot k^2 \cdot \eps^{-2} \leq  |I|$. The running time of Step~\ref{step:A} is $|I|^{O(1)}$ by Lemma~\ref{lem:CA}. Each iteration of the \textbf{for} loop in Step~\ref{step:REPforr} can be computed in time $\Tilde{O}(\ell^{(k-1) \cdot \ell} \cdot k) \cdot |I|^{O(1)}$, by Lemma~\ref{lem:mainMatroid}. Hence, as we have $|D| = |D(I,\eps)|$ iterations of the  \textbf{for} loop, the running time of the algorithm is bounded by 
$$
|D| \cdot \Tilde{O}(\ell^{(k-1) \cdot \ell} \cdot k) \cdot |I|^{O(1)} \leq
(2 \ell \cdot k \cdot \eps^{-2} +1) \cdot  \Tilde{O}(\ell^{(k-1) \cdot \ell} \cdot k)
\cdot |I|^{O(1)}
= \Tilde{O}\left(\ell^{(k-1) \cdot \ell +1} \cdot k^2 \cdot \eps^{-2} \right) \cdot |I|^{O(1)}.
$$ 
The first inequality follows from \eqref{eq:ing} and \eqref{eq:upper_bound_D}. 
As in this case $\ell^{(k-1) \cdot \ell } \cdot k^2 \cdot \eps^{-2} \leq  |I|$,
we have the desired running time. 

For the cardinality of $R$, note that by Lemma~\ref{lem:CA} $\OPT(I) \geq \alpha \geq \frac{\OPT(I)}{2 \ell}$. Thus, by Lemma~\ref{lem:mainMatroid}, for all $r \in D$,  $\textsf{ExSet}(I,\eps,\alpha,r)$ is an exchange set satisfying 
$\left|  \textsf{ExSet}(I,\eps,\alpha,r) \right| = \Tilde{O}(\ell^{(k-1) \cdot \ell} \cdot k)$. Then, 

\begin{equation*}
|R| \leq |D| \cdot \Tilde{O}(\ell^{(k-1) \cdot \ell} \cdot k) \leq 
(2 \ell \cdot k \cdot \eps^{-2} +1) \cdot \Tilde{O}(\ell^{(k-1) \cdot \ell} \cdot k)
= \Tilde{O}\left(\ell^{(k-1) \cdot \ell +1} \cdot k^2 \cdot \eps^{-2} \right). 
\end{equation*}  
\noindent
The second inequality follows from \eqref{eq:ing} and \eqref{eq:upper_bound_D}. 

To conclude, we show that $R$ is a representative set. By Lemma~\ref{lem:mainMatroid}, for all $r \in D$, it holds that $\textsf{ExSet}(I,\eps,\alpha,r)$ is an exchange set for $I,\eps,\alpha$, and $r$. Therefore, 
$R \cap {\ck}_r(\alpha)$ is an exchange set for $I,\eps,\alpha$, for all $r \in  D$. Hence, by Lemma~\ref{lem:sufficientRep}, $R$ is a representative set of $I$ and~$\eps$. 
\end{proof}

 \noindent{\bf Proof of Lemma~\ref{lem:Main}:} The statement of the lemma follows  from Lemma~\ref{lem:main}. \qed

\section{An FPT Approximation Scheme}
\label{sec:FPAS}
In this section we use the representative set constructed by Algorithm~\ref{alg:findRep}
to obtain an FPAS for BM. For the discussion below, fix a BM  instance $I = (E, \cm, c,p, B,k,\ell)$ and an error parameter $0<\eps <\frac{1}{2}$. Given the representative set $R$ for $I$ and $\eps$ output by
algorithm \textsf{RepSet}, we derive an FPAS by exhaustive enumeration over all solutions of $I$ within $R$. The pseudocode of our FPAS is given in Algorithm~\ref{alg:EPTAS}. 

\begin{algorithm}[h]
	\caption{$\textsf{FPAS}(I = (E, \cm, c,p, B,k,\ell),\eps)$}
	\label{alg:EPTAS}
	
	
	\SetKwInOut{Input}{input}
	
	\SetKwInOut{Output}{output}
	
	\Input{A BM instance $I$ and an error parameter $0<\eps<\frac{1}{2}$.}
	
	\Output{A solution for $I$.}

	Initialize an empty solution $A \leftarrow \emptyset$.\label{step:init}
	

	Construct $R \leftarrow \textsf{RepSet} (I,\eps)$.\label{step:rep}

	\For{$F \subseteq R \textnormal{ s.t. } |F| \leq k \textnormal{ and } F \textnormal{ is a solution of } I $ \label{step:for}}{

		

		\If{$p\left(F\right) > p(A)$\label{step:iff}}{
			
			Update $A \leftarrow F$\label{step:update}
			
		}

	}
	
	Return $A$.\label{step:retA}
\end{algorithm}

\begin{lemma}
	\label{thm:EPTAS}
	Given a \textnormal{BM} instance $I = (E, \cm, c,p, B,k,\ell)$ and $0<\eps<\frac{1}{2}$, Algorithm~\ref{alg:EPTAS} returns in time $|I|^{O(1)} \cdot  \Tilde{O} \left( \ell^{k^2 \cdot \ell} \cdot k^{2k} \cdot \eps^{-2k} \right)$ a solution for $I$ of profit at least $(1-2\eps) \cdot \OPT(I)$. 
\end{lemma}

We can now prove our main result.

 \noindent{\bf Proof of Lemma~\ref{thm:FPAS}:} The proof follows 
from Lemma~\ref{thm:EPTAS} by using in Algorithm~\ref{alg:EPTAS} an error parameter 
 $\eps' = \frac{\eps}{2}$. \qed

For the proof of Lemma~\ref{thm:EPTAS}, we use the next auxiliary lemmas.  

\begin{lemma}
	\label{thm:aux1}
	Given a \textnormal{BM} instance $I = (E, \cm, c,p, B,k,\ell)$ and $0<\eps<\frac{1}{2}$, Algorithm~\ref{alg:EPTAS} returns a solution for $I$ of profit at least $(1-2\eps) \cdot \OPT(I)$. 
\end{lemma}

\begin{proof}
	By Lemma~\ref{lem:main}, it holds that $R = \textsf{RepSet}(I,\eps)$ is a representative set of $I$ and $\eps$. Therefore, by Definition~\ref{def:REP}, there is a solution $S$ for $I$ such that $S \subseteq R$, and \begin{equation}
		\label{eq:profitS}
		p\left(S\right) \geq (1-2\eps) \cdot \OPT(I).
	\end{equation} Since $S$ is a solution for $I$, it follows that $S \in \cm_{k}$ and therefore $|S| \leq k$. Thus, 
there is an iteration of Step~\ref{step:for} in which $F = S$, and therefore the set $A$ returned by the algorithm satisfies  $p(A) \geq p(S) \geq (1-2\eps) \cdot \OPT(I)$. Also, the set $A$ returned by the algorithm must be a solution for $I$: If $A = \emptyset$ the claim trivially follows since $\emptyset$ is a solution for $I$.
Otherwise, the value of $A$ has been updated in Step~\ref{step:update} of Algorithm~\ref{alg:EPTAS} to be some set $F \subseteq R$, but this step is reached only if $F$ is a solution for $I$.  
\end{proof}

\begin{lemma}
	\label{thm:aux2}
	Given a \textnormal{BM} instance $I = (E, \cm, c,p, B,k,\ell)$ and $0<\eps<\frac{1}{2}$, the running time of Algorithm~\ref{alg:EPTAS} is $|I|^{O(1)} \cdot  \Tilde{O} \left( \ell^{k^2 \cdot \ell} \cdot k^{2k} \cdot \eps^{-2k} \right)$.
\end{lemma}
\begin{proof}
Let $$W' = \big\{F \subseteq R~\big|~ F \in \cm_k, c(F) \leq B\big\}$$ be the solutions considered in Step~\ref{step:for} of Algorithm~\ref{alg:EPTAS}, and let $$W = \big\{F \subseteq R~\big|~ |F| \leq k\big\}.$$ Observe that the number of iterations of Step~\ref{step:for} of Algorithm~\ref{alg:EPTAS} is bounded by $|W|$, since $W' \subseteq W$ and for each $F \in W$ we can verify in polynomial time if $F \in W'$. Thus, it suffices to upper bound $W$. 
	
	By a simple counting argument, we have that 
\begin{equation}
			\label{eq:subR}
			\begin{aligned}
				|W| \leq{} & 
				\left( 
				|R|+1\right)^{k} \\
				\leq{} & 
				\Tilde{O} \left( \left( \ell^{(k-1) \cdot \ell +1} \cdot k^2 \cdot \eps^{-2} \right)^k \right) \\
				={} & \Tilde{O} \left( \ell^{k^2 \cdot \ell} \cdot k^{2k} \cdot \eps^{-2k}
				\right)				
			\end{aligned}
		\end{equation} 
		The first equality follows from Lemma~\ref{lem:main}. 
		Hence, by \eqref{eq:subR}, the number of iterations of the {\bf for} loop  in Step~\ref{step:for} is bounded by $ \Tilde{O} \left( \ell^{k^2 \cdot \ell} \cdot k^{2k} \cdot \eps^{-2k} \right)$. In addition, the running time of each iteration is at most $|I|^{O(1)}$. Finally, the running time of the steps outside the {\bf for} loop is $|I|^{O(1)}$, by
		Lemma~\ref{lem:main}. Hence, the running time of Algorithm~\ref{alg:EPTAS} can be bounded by $|I|^{O(1)} \cdot  \Tilde{O} \left( \ell^{k^2 \cdot \ell} \cdot k^{2k} \cdot \eps^{-2k} \right) $.  
	\end{proof} 

\noindent{\bf Proof of Lemma~\ref{thm:EPTAS}:} The proof follows from Lemmas~\ref{thm:aux1} and~\ref{thm:aux2}. \qed

\section{Hardness Results}

\label{sec:hard}
	
In this section we prove Lemma~\ref{lem:hard} and Lemma~\ref{lem:DM}. In the proof of Lemma~\ref{lem:hard}, we use a reduction from the $k$-{\em subset sum (KSS)} problem. The input for KSS is a set $X = \{x_1, \ldots, x_n\}$ of strictly positive integers and two positive integers $T,k>0$. We need to decide if there is a subset $S \subseteq [n], |S| = k$ such that $\sum_{i \in S} x_i = T$, where the problem is parameterized by $k$. KSS is known to be W[1]-hard \cite{downey1995fixed}. 
	
	\noindent{\bf Proof of Lemma~\ref{lem:hard}:} Let $U$ be a KSS instance with the set of numbers $E = [n]$, target value $T$, and $k$. We define the following BM instance $I = (E, \cm, c,p, B,k,\ell)$,.
	 \begin{enumerate}
		\item $\cm$ is a $1$-matchoid $\cm = \{(E,\cI)\}$ such that $\cI = 2^E$. That is, $\cm$ is a single uniform matroid whose independent sets are all possible subsets of $E$.
		\item For any $i \in E = [n]$ define $c(i) = p(i) = x_i+2 \cdot \sum_{j \in [n]} x_j$. 
		\item Define the budget as $B = T+2 k \cdot \sum_{j \in [n]} x_j$.
	\end{enumerate} 
	\begin{claim}
		\label{claim:1}
		If there is a solution for $U$ then there is a solution for $I$ of profit $B$. 
	\end{claim}
	\begin{claimproof}
		Let $S \subseteq [n], |S| = k$ such that $\sum_{i \in S} x_i = T$. 
		Then, $$c(S) = p(S) = \sum_{i \in S} \left( x_i+2 \cdot \sum_{j \in [n]} x_j  \right) = T+|S| \cdot 2 \cdot \sum_{j \in [n]} x_j = T+2 k\cdot \sum_{j \in [n]} x_j = B.$$ By the above, and as $S \in \cm_k$, $S$ is also a solution for $I$ of profit exactly $B$. 
	\end{claimproof}
	\begin{claim}
		\label{claim:2}
		If there is a solution for $I$ of profit at least $B$ then there is a solution for $U$. 
	\end{claim}
	\begin{claimproof}
		Let $F$ be a solution for $I$ of profit at least $B$. Then, $p(F) = c(F) \leq B$, since $F$ satisfies the budget constraint. As $p(F) \geq B$, we conclude that 
\begin{equation}
\label{eq:profit_F_equals_B}		
p(F) = c(F) = B.
\end{equation}
We now show that $F$ is also a solution for~$U$. First, assume towards contradiction that $|F| \neq k$. If $|F|< k$ then 
$$
p(F) = \sum_{i \in F} x_i+|F| \cdot 2 \cdot \sum_{j \in [n]} x_j \leq \sum_{i \in F} x_i+(k-1) \cdot 2 \cdot \sum_{j \in [n]} x_j \leq 2 k\cdot \sum_{j \in [n]} x_j < B.
$$
We reach a contradiction to \eqref{eq:profit_F_equals_B}. 
%
Since $F$ is a solution for $I$ it holds that $F \in \cm_{k}$; thus, $|F| \leq k$. By the above, $|F| = k$.  Therefore,
		$$
		\sum_{i \in F} x_i = c(F) - |F| \cdot 2 \cdot \sum_{j \in [n]} x_j = c(F) -  2 k \cdot \sum_{j \in [n]} x_j = B - 2 k \cdot \sum_{j \in [n]} x_j = T.
		$$	
	\end{claimproof}
	By Claims~\ref{claim:1} and~\ref{claim:2}, there is a solution for $U$ if and only if there is a solution for $I$ of profit at least $B$. Furthermore, the construction of $I$ can be done in polynomial time in the encoding size of $U$. Hence, an FPT algorithm which finds an optimal solution for $I$ can decide the instance $U$ in FPT time.
	As KSS is known to be W[1]-hard \cite{downey1995fixed}, we conclude that BM is also W[1]-hard.  \qed
	
In the proof of \Cref{lem:DM} we use a lower bound on the kernel size of 	
 {\em Perfect $\ell$-Dimensional Matching} ($\ell$-PDM), due to Dell and Marx \cite{dell2012kernelization,dell2018kernelization_arxiv}. The input for the problem consists of the finite sets $U_1,\ldots U_{\ell}$ and  $E\subseteq U_1\times \ldots \times U_{\ell}$. Also, we have an  
$\ell$-dimensional matching constraint $(E,\cI)$ to which we refer as the associated set system of the instance (i.e., $\cI$ contains all subsets $S \subseteq E$ such that for any two distinct tuples $(e_1,\ldots, e_{\ell}), (f_1,\ldots, f_{\ell}) \in S$ and every $i \in [\ell]$ it holds that $e_i \neq f_i$). The instance is associated also  with the parameter $k=\frac{n}{\ell}$, where $n=\sum_{j=1}^{\ell} |U_{\ell}|$. We refer to $|E|$ as the {\em number of tuples} in the instance. 
	The objective is to find  $S\in \cI$ such that $|S| = k$. Let $J=(U_1, \ldots , U_{\ell},E)$ denote an instance of $\ell$-PDM
	We say $J$ is a ``yes'' instance if such a set $S$ exists; otherwise, $J$ is a ``no'' instance. Observe that the parameter~$k$ is set such that  if $S\in \cI$ and $|S| = k$ then every element in $U_1\cup \ldots \cup U_{\ell}$ appears in exactly one of the tuples in $S$. 
	
	\begin{lemma}[Theorem 1.2 cf. \cite{dell2018kernelization_arxiv}]
		\label{lem:dell}
	Let $\ell \geq3$ and $\eps>0$. If  $ \textnormal{coNP}\not\subseteq \textnormal{NP} / \textnormal{poly}$ then $\ell$-\textnormal{PDM} does not have a kernel in which the number of tuples is $O(k^{\ell-\eps})$.
		\end{lemma}

	\noindent{\bf Proof of Lemma~\ref{lem:DM}:} 
	Assume $ \textnormal{coNP}\not\subseteq \textnormal{NP} / \textnormal{poly}$. Furthermore, 
	assume towards a contradiction that there is a function $f:\mathbb{N} \rightarrow \mathbb{N}$, constants $c_1,c_2$, where $c_2-c_1<0$, and an algorithm~$\cA$ that, given a \textnormal{BM} instance $I = (E, \cm, c,p, B,k,\ell)$  and $0<\eps<\frac{1}{2}$, finds in time $|I|^{O(1)}$ a representative set of $I$ and $\eps$ of size $O \left( f(\ell)  \cdot k^{\ell-c_1} \cdot \frac{1}{\eps^{c_2}}\right)$. We use $\cA$ to construct a kernel for $3$-PDM. 
	
	Consider the following kernelization algorithm for $3$-PDM. Let $J=(U_1,U_2,U_3,E)$ be the $3$-PDM  input instance. Define $n=|U_1|+|U_2|+|U_3|$, $\ell =3$, and  $k=\frac{n}{\ell}$. 
	Furthermore, let $(E,\cI)$ be the set system associated with the instance, and let $\cm$ be an $\ell$-matchoid representing the set system $(E,\cI)$.
	Run~$\cA$ on the BM instance $I=(E,\cm,\c,p,B,k,\ell)$ with $\eps =\frac{1}{3k}$,  where $c(e)=p(e)=1$ for all $e\in E$ and $B=k$.  Let $R\subseteq E$ be the output of~$\cA$. Return the $3$-PDM instance  $J'=(U_1,U_2,U_3,R)$.
		
Since~$\cA$ runs in polynomial time, the above algorithm runs in polynomial time as well.  Moreover, as $k= \frac{n}{3}$ and $R\subseteq E$, it follows that the returned instance can be encoded using $O(k^4)$ bits. Let $(R,\cI')$ be the set system associated with $J'$.
Since $R\subseteq E$, it follows that $\cI'\subseteq \cI$. Hence, if there is $S\in \cI'$ such that $|S|=k$, then $S\in \cI$ as well. That is,  
if $J'$ is a ``yes'' instance, so is $J$. 

For the other direction, assume that $J$ is a ``yes'' instance. That is, there is $S\in \cI$ such that $|S|=k$. Then $S$ is a solution for the BM instance $I$ (observe that $c(S)=|S|=k=B$). Therefore, as $R$ is a representative set of $I$ and $\eps=\frac{1}{3k}$, there is a solution $T$ for $I$ such that $T\subseteq R$, and 
$$
p(T) \geq (1-2\eps)\cdot \OPT(I)\geq \left(1-2\eps\right)\cdot p(S)
 = \left(1-\frac{2}{3k}\right)\cdot p(S)=  \left(1-\frac{2}{3k}\right)\cdot k = k-\frac{2}{3}.$$
Since the profits are integral we have that $|T|=p(T)\geq k$. Furthermore $|T|\leq k$ (since $T$ is a solution for $I$), and thus $|T|=k$. Since $T\in \cI$ (as $T$ is a  solution for $I$) and $T\subseteq R$, it trivially holds that $T\in \cI'$. That is, $T\in \cI'$ and $|T|=k$. Hence, $J'$ is a ``yes'' instance. We have showed that the above procedure is indeed a kernelization for $3$-PDM. 

Now, consider the size of $R$.
Since $\cA$ returns a representative set of size $O \left( f(\ell)  \cdot k^{\ell-c_1} \cdot \frac{1}{\eps^{c_2}}\right)$ it follows that 
$$
|R| = O\left(  f(3)  \cdot k^{3-c_1} \cdot (3k)^{c_2} \right) = O\left( k^{3-c_1+c_2} \right).
$$
As $c_2-c_1<0$, we have a contradiction to \Cref{lem:dell}. Thus, for any function $f:\mathbb{N} \rightarrow \mathbb{N}$ and  constants $c_1,c_2$ satisfying  $c_2-c_1<0$, there is no  algorithm which finds for a given \textnormal{BM} instance $I = (E, \cm, c,p, B,k,\ell)$  and $0<\eps<\frac{1}{2}$ a representative set of $I$ and $\eps$ of size $O \left( f(\ell)  \cdot k^{\ell-c_1} \cdot \frac{1}{\eps^{c_2}}\right)$
in time $|I|^{O(1)}$.
\qed

\section{A Polynomial Time $\frac{1}{2\cdot \ell}$-Approximation for BM}
\label{sec:proofs}

In this section we prove \Cref{lem:CA}. The proof combines an existing approximation algorithm for the unbudgeted version of BM \cite{jenkyns1976efficacy,jukna2011extremal} with the Lagrangian relaxation technique of \cite{kulik2021lagrangian}.  As the results in \cite{jenkyns1976efficacy,jukna2011extremal}  are presented in the context of {\em $\ell$-extendible } set systems, we first define these systems and use a simple argument to show that such systems are generalizations of matchoids. We refer the reader to \cite{feldman2011improved} for further details about $\ell$-extendible systems. 
 \begin{definition}
	\label{def:EX}
	Given a finite set $E$, $\cI \subseteq 2^E$, and $\ell \in \mathbb{N}$, we say that $(E,\cI)$ is an {\em $\ell$-extendible system} if for every $S \in \cI$ and $e \in E \setminus S$ there is $T \subseteq S$, where $|T| \leq \ell$, such that $(S \setminus T) \cup \{\ell\} \in \cI$. 
\end{definition}

The next lemma shows that an \lm is in fact an $\ell$-extendible set system.

\begin{lemma}
	\label{lem:EX}
	For any $\ell \in \mathbb{N}_{>0}$ and an $\ell$-Matchoid  $\cm = \left\{  M_i = (E_i, \cI_i) \right\}_{i \in [s]}$ on a set $E$, it holds that $(E,\cI(\cm))$ is an $\ell$-extendible set system. 
\end{lemma}

\begin{proof}
	Let $S \in \cI(\cm)$ and $e \in E \setminus S$. As $\cm$ is an $\ell$-matchoid, there is $H \subseteq [s]$ of cardinality $|H| \leq \ell$ such that for all $i \in [s] \setminus H$ it holds that $e \notin E_i$ and for all $i \in H$ it holds that $e \in E_i$. Since for all $i \in H$ it holds that $(E_i,\cI_i)$ is a matroid, either $(S \cap E_i) \cup \{e\} \in \cI_i$, or there is $a_i \in S \cap E_i$ such that  $((S \cap E_i) \setminus \{a_i\}) \cup \{e\} \in \cI_i$ (this follows by repeatedly adding elements from $S \cap E_i$ to $\{e\}$ using the exchange property of the matroid $(E_i,\cI_i)$). Let $L = \{i \in H~|~(S \cap E_i) \cup \{e\} \notin \cI_i\}$. Then, there are $|L|$ elements $T = \{a_{i}\}_{i \in L}$ such that for all $i \in L$ it holds that $((S \cap E_i) \setminus \{a_i\}) \cup \{e\} \in \cI_i$ and for all $i \in H \setminus L$ it holds that $(S \cap E_i) \cup \{e\} \in \cI_i$. Thus, it follows that $(S \setminus T) \cup \{e\} \in \cI(\cm)$ by the definition of a matchoid. Since $|T| = |L| \leq |H| \leq \ell$, we have the statement of the lemma. 
\end{proof}

%

\noindent {\bf Proof of Lemma~\ref{lem:CA}:} Consider the BM problem with no budget constraint (equivalently, $B>c(E)$) that we call the {\em maximum weight matchoid maximization (MWM)} problem. By Lemma~\ref{lem:EX}, MWM is a special case of the {\em maximum weight $\ell$-extendible system maximization} problem, which admits $\frac{1}{\ell}$-approximation \cite{jenkyns1976efficacy,jukna2011extremal}.\footnote{The algorithm of \cite{jenkyns1976efficacy} can be applied also in the
more general setting of $\ell$-systems. For more details on such set systems, see, e.g., \cite{feldman2011improved}.} Therefore, 
using a technique of~\cite{kulik2021lagrangian}, we have the following. There is an algorithm that, given some $\eps>0$, returns a solution for the BM instance $I$ of profit at least $\left(  \frac{\frac{1}{\ell}}{\frac{1}{\ell}+1} -\eps \right) \cdot \OPT(I)$, and whose running time is $|I|^{O(1)} \cdot O(\log(\eps^{-1}))$. Now, we can set $\eps = \frac{\frac{1}{\ell}}{\frac{1}{\ell}+1}-\frac{1}{2\ell}$; then, the above algorithm has a running time $|I|^{O(1)}$, since $\eps^{-1}$ is polynomial in $\ell$ and $\ell \leq |I|$. Moreover, the algorithm returns a solution $S$ for $I$, such that $$\OPT(I) \geq p(S) \geq \left(  \frac{\frac{1}{\ell}}{\frac{1}{\ell}+1} -\eps \right) \cdot \OPT(I)  = \frac{1}{2\ell} \cdot \OPT(I).$$
To conclude, we define the algorithm $\textnormal{\textsf{ApproxBM}}$ which returns $\alpha = p(S)$. By the above discussion, $\OPT(I) \geq \alpha \geq \frac{\OPT(I)}{2\ell}$, and the running time of $\textnormal{\textsf{ApproxBM}}$ is 
$|I|^{O(1)}$.  \qed

\section{Discussion}
\label{sec:discussion}

In this paper we present an FPT-approximation scheme (FPAS) for the budgeted $\ell$-matchoid  problem (BM). As special cases, this yields FPAS for 
 the budgeted $\ell$-dimensional matching problem  (BDM) and the budgeted $\ell$-matroid intersection problem (BMI). While the unbudgeted version of BM has been studied earlier from parameterized viewpoint, the budgeted version is studied here for the first time. 

 We show that BM parameterized by the solution size is $W[1]$-hard already with a degenerate matroid constraint (Lemma~\ref{lem:hard}); thus, an exact FPT time algorithm is unlikely to exist. Furthermore, the special case of unbudgeted $\ell$-dimensional matching problem is APX-hard, already for $\ell=3$, implying that  
 PTAS for this problem is also unlikely to exist.
 These hardness results motivated the  development of an FPT-approximation scheme for BM.
 
Our FPAS relies on the notion of representative set $-$ a small cardinality  subset of the ground set of the original instance which  preserves the optimum value up to a small factor. We note that representative sets are not {\em lossy kernels} \cite{lokshtanov2017lossy} as BM is defined in an oracle model;  thus, the definitions of kernels or lossy kernels do not apply to our problem. Nevertheless, for some variants of BM in which the input is given explicitly (for instance, this is possible for BDM) our construction of  representative sets can be used to obtain an approximate kernelization scheme. 

Our results also include a lower bound on the minimum possible size of a representative set for BM which can be computed in polynomial time (\Cref{lem:DM}). The lower bound is based on the special case of the budgeted $\ell$-dimensional matching problem (BDM). We note that there is a significant gap between the size of the representative sets found in this paper and the lower bound. This suggests the following questions for future work.
\begin{itemize}
	\item Is there a representative set for the special case of BDM whose size matches the lower bound given in \Cref{lem:DM}?
	\item Can the generic structure of $\ell$-matchoids be used to derive an improved lower bound on the size of a representative set for general BM instances?
\end{itemize}

The budgeted $\ell$-matchoid problem can be naturally  generalized to the $d$-budgeted $\ell$-matchoid problem ($d$-BM). In the $d$-budgeted version, both the costs and the budget are replaced by $d$-dimensional vectors, for some constant $d\geq 2$.  
A subset of elements is feasible if it is an independent set of the $\ell$-matchoid, and the total cost of the elements in each dimension is bounded by the budget in this dimension. The problem is a generalization of the $d$-dimensional knapsack problem ($d$-KP), the special case of $d$-BM in which the feasible sets of the matchoid are all subsets of $E$. A PTAS for $d$-KP  was first given in~\cite{FC84}, and the existence of an {\em efficient}
polynomial time approximation scheme was ruled out in \cite{kulik2010there}.  
PTASs for the special cases of $d$-BM in which the matchoid is a single matroid, matroid intesection or a matching constraint were given in \cite{chekuri2011multi,grandoni2010approximation}. 
It is likely that the  lower bound in~\cite{kulik2010there} can be used also to rule out the existence of an FPAS for $d$-BM. However, the question whether $d$-BM admits a 
$(1-\eps)$-approximation in time $O\left( f(k+\ell) \cdot n^{g(\eps)}\right)$, for some functions $f$ and $g$, remains open. 
\bibliography{budgeted}

\begin{thebibliography}{1}

\bibitem{BBGS11}
Andr{\'e} Berger, Vincenzo Bonifaci, Fabrizio Grandoni, and Guido Sch{\"a}fer.
\newblock Budgeted matching and budgeted matroid intersection via the gasoline
  puzzle.
\newblock {\em Mathematical Programming}, 128(1):355--372, 2011.

\bibitem{cormen2022introduction}
Thomas~H Cormen, Charles~E Leiserson, Ronald~L Rivest, and Clifford Stein.
\newblock {\em Introduction to algorithms}.
\newblock MIT press, 2022.

\bibitem{doron2022eptas}
Ilan Doron-Arad, Ariel Kulik, and Hadas Shachnai.
\newblock An eptas for budgeted matroid independent set.
\newblock {\em arXiv preprint arXiv:2209.04654}, 2022.

\bibitem{grandoni2010approximation}
Fabrizio Grandoni and Rico Zenklusen.
\newblock Approximation schemes for multi-budgeted independence systems.
\newblock In {\em European Symposium on Algorithms}, pages 536--548. Springer,
  2010.

\bibitem{sahni1975approximate}
Sartaj Sahni.
\newblock Approximate algorithms for the 0/1 knapsack problem.
\newblock {\em Journal of the ACM (JACM)}, 22(1):115--124, 1975.

\bibitem{schrijver2003combinatorial}
Alexander Schrijver.
\newblock {\em Combinatorial optimization: polyhedra and efficiency},
  volume~24.
\newblock Springer, 2003.

\end{thebibliography}


\begin{thebibliography}{10}

\bibitem{BBGS11}
Andr{\'e} Berger, Vincenzo Bonifaci, Fabrizio Grandoni, and Guido Sch{\"a}fer.
\newblock Budgeted matching and budgeted matroid intersection via the gasoline
  puzzle.
\newblock {\em Mathematical Programming}, 128(1):355--372, 2011.

\bibitem{bjorklund2017narrow}
Andreas Bj{\"o}rklund, Thore Husfeldt, Petteri Kaski, and Mikko Koivisto.
\newblock Narrow sieves for parameterized paths and packings.
\newblock {\em Journal of Computer and System Sciences}, 87:119--139, 2017.

\bibitem{chekuri2011multi}
Chandra Chekuri, Jan Vondr{\'a}k, and Rico Zenklusen.
\newblock Multi-budgeted matchings and matroid intersection via dependent
  rounding.
\newblock In {\em Proceedings of the twenty-second annual ACM-SIAM symposium on
  Discrete Algorithms}, pages 1080--1097. SIAM, 2011.

\bibitem{chen2011improved}
Jianer Chen, Qilong Feng, Yang Liu, Songjian Lu, and Jianxin Wang.
\newblock Improved deterministic algorithms for weighted matching and packing
  problems.
\newblock {\em Theoretical computer science}, 412(23):2503--2512, 2011.

\bibitem{dell2012kernelization}
Holger Dell and D{\'a}niel Marx.
\newblock Kernelization of packing problems.
\newblock In {\em Proceedings of the twenty-third annual ACM-SIAM symposium on
  Discrete Algorithms}, pages 68--81. SIAM, 2012.

\bibitem{dell2018kernelization_arxiv}
Holger Dell and D{\'a}niel Marx.
\newblock Kernelization of packing problems.
\newblock {\em arXiv preprint arXiv:1812.03155}, 2018.

\bibitem{doron2023eptas}
Ilan Doron-Arad, Ariel Kulik, and Hadas Shachnai.
\newblock An {EPTAS} for budgeted matching and budgeted matroid intersection.
\newblock {\em To appear in ICALP}, 2023.

\bibitem{DKS23}
Ilan Doron-Arad, Ariel Kulik, and Hadas Shachnai.
\newblock An {EPTAS} for budgeted matroid independent set.
\newblock In {\em Symposium on Simplicity in Algorithms (SOSA)}, pages 69--83,
  2023.

\bibitem{downey1995fixed}
Rod~G Downey and Michael~R Fellows.
\newblock Fixed-parameter tractability and completeness ii: On completeness for
  w [1].
\newblock {\em Theoretical Computer Science}, 141(1-2):109--131, 1995.

\bibitem{feldman2011improved}
Moran Feldman, Joseph Naor, Roy Schwartz, and Justin Ward.
\newblock Improved approximations for k-exchange systems.
\newblock In {\em Algorithms--ESA 2011: 19th Annual European Symposium,
  Saarbr{\"u}cken, Germany, September 5-9, 2011. Proceedings 19}, pages
  784--798. Springer, 2011.

\bibitem{FLM20}
Andreas~Emil Feldmann, Euiwoong Lee, and Pasin Manurangsi.
\newblock A survey on approximation in parameterized complexity: Hardness and
  algorithms.
\newblock {\em Algorithms}, 13(6):146, 2020.

\bibitem{fomin2016efficient}
Fedor~V Fomin, Daniel Lokshtanov, Fahad Panolan, and Saket Saurabh.
\newblock Efficient computation of representative families with applications in
  parameterized and exact algorithms.
\newblock {\em Journal of the ACM (JACM)}, 63(4):1--60, 2016.

\bibitem{fomin2019kernelization}
Fedor~V Fomin, Daniel Lokshtanov, Saket Saurabh, and Meirav Zehavi.
\newblock {\em Kernelization: theory of parameterized preprocessing}.
\newblock Cambridge University Press, 2019.

\bibitem{FC84}
Alan~M Frieze, Michael~RB Clarke, et~al.
\newblock Approximation algorithms for the m-dimensional 0-1 knapsack problem:
  worst-case and probabilistic analyses.
\newblock {\em European Journal of Operational Research}, 15(1):100--109, 1984.

\bibitem{garey1979computers}
Michael~R Garey and David~S Johnson.
\newblock {\em Computers and intractability}, volume 174.
\newblock freeman San Francisco, 1979.

\bibitem{goyal2015deterministic}
Prachi Goyal, Neeldhara Misra, Fahad Panolan, and Meirav Zehavi.
\newblock Deterministic algorithms for matching and packing problems based on
  representative sets.
\newblock {\em SIAM Journal on Discrete Mathematics}, 29(4):1815--1836, 2015.

\bibitem{grandoni2010approximation}
Fabrizio Grandoni and Rico Zenklusen.
\newblock Approximation schemes for multi-budgeted independence systems.
\newblock In {\em European Symposium on Algorithms}, pages 536--548. Springer,
  2010.

\bibitem{huang2023fpt}
Chien-Chung Huang and Justin Ward.
\newblock {FPT}-algorithms for the l-matchoid problem with a coverage
  objective.
\newblock {\em SIAM Journal on Discrete Mathematics}, 2023.

\bibitem{jenkyns1976efficacy}
Th~Jenkyns.
\newblock The efficacy of the "greedy" algorithm.
\newblock In {\em Proc. 7th Southeastern Conf. on Combinatorics, Graph Theory
  and Computing}, pages 341--350, 1976.

\bibitem{jensen1982complexity}
Per~M Jensen and Bernhard Korte.
\newblock Complexity of matroid property algorithms.
\newblock {\em SIAM Journal on Computing}, 11(1):184--190, 1982.

\bibitem{jukna2011extremal}
Stasys Jukna.
\newblock {\em Extremal combinatorics: with applications in computer science},
  volume 571.
\newblock Springer, 2011.

\bibitem{kann1991maximum}
Viggo Kann.
\newblock Maximum bounded 3-dimensional matching is max snp-complete.
\newblock {\em Information Processing Letters}, 37(1):27--35, 1991.

\bibitem{koutis2009limits}
Ioannis Koutis and Ryan Williams.
\newblock Limits and applications of group algebras for parameterized problems.
\newblock In {\em Automata, Languages and Programming: 36th International
  Colloquium, ICALP 2009, Rhodes, Greece, July 5-12, 2009, Proceedings, Part I
  36}, pages 653--664. Springer, 2009.

\bibitem{kulik2010there}
Ariel Kulik and Hadas Shachnai.
\newblock There is no {EPTAS} for two-dimensional knapsack.
\newblock {\em Information Processing Letters}, 110(16):707--710, 2010.

\bibitem{kulik2021lagrangian}
Ariel Kulik, Hadas Shachnai, and Gal Tamir.
\newblock On lagrangian relaxation for constrained maximization and
  reoptimization problems.
\newblock {\em Discrete Applied Mathematics}, 296:164--178, 2021.

\bibitem{lee2009non}
Jon Lee, Vahab~S Mirrokni, Viswanath Nagarajan, and Maxim Sviridenko.
\newblock Non-monotone submodular maximization under matroid and knapsack
  constraints.
\newblock In {\em Proceedings of the forty-first annual ACM symposium on Theory
  of computing}, pages 323--332, 2009.

\bibitem{lokshtanov2018deterministic}
Daniel Lokshtanov, Pranabendu Misra, Fahad Panolan, and Saket Saurabh.
\newblock Deterministic truncation of linear matroids.
\newblock {\em ACM Transactions on Algorithms (TALG)}, 14(2):1--20, 2018.

\bibitem{lokshtanov2017lossy}
Daniel Lokshtanov, Fahad Panolan, MS~Ramanujan, and Saket Saurabh.
\newblock Lossy kernelization.
\newblock In {\em Proceedings of the 49th Annual ACM SIGACT Symposium on Theory
  of Computing}, pages 224--237, 2017.

\bibitem{lovasz1980matroid}
L{\'a}szl{\'o} Lov{\'a}sz.
\newblock Matroid matching and some applications.
\newblock {\em Journal of Combinatorial Theory, Series B}, 28(2):208--236,
  1980.

\bibitem{martello1990knapsack}
Silvano Martello and Paolo Toth.
\newblock {\em Knapsack problems: algorithms and computer implementations}.
\newblock John Wiley \& Sons, Inc., 1990.

\bibitem{marx2008parameterized}
D{\'a}niel Marx.
\newblock Parameterized complexity and approximation algorithms.
\newblock {\em The Computer Journal}, 51(1):60--78, 2008.

\bibitem{marx2009parameterized}
D{\'a}niel Marx.
\newblock A parameterized view on matroid optimization problems.
\newblock {\em Theoretical Computer Science}, 410(44):4471--4479, 2009.

\bibitem{ravi1996constrained}
Ram Ravi and Michel~X Goemans.
\newblock The constrained minimum spanning tree problem.
\newblock In {\em Scandinavian Workshop on Algorithm Theory}, pages 66--75.
  Springer, 1996.

\bibitem{schrijver2003combinatorial}
Alexander Schrijver et~al.
\newblock {\em Combinatorial optimization: polyhedra and efficiency},
  volume~24.
\newblock Springer, 2003.

\end{thebibliography}
%

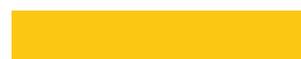
\end{document}